\documentclass[aps,prl,superscriptaddress,twocolumn,nofootinbib,a4paper,amsfonts,amssymb,amsmath,floatfix]{revtex4-2}
\pdfoutput=1
\usepackage[utf8]{inputenc}
\usepackage[english]{babel}
\usepackage[colorlinks=true,allcolors=magenta,breaklinks=true]{hyperref}
\usepackage{subcaption}
\usepackage{graphicx}
\usepackage{dsfont}
\usepackage{soul}
\usepackage{microtype}
\usepackage{mathtools}
\usepackage{amsthm}
\usepackage{tikz}
\usetikzlibrary{calc,shapes,patterns,decorations.pathmorphing,decorations.pathreplacing,calligraphy,svg.path,arrows.meta}

\newcommand{\ket}[1]{|#1\rangle}         
\newcommand{\bra}[1]{\langle#1|}         
\newcommand{\cpm}{\operatorname{CP}}     
\newcommand{\cptp}{\operatorname{CPTP}}  
\DeclareMathOperator{\id}{\mathds{1}}    
\DeclareMathOperator{\Tr}{Tr}            
\DeclareMathOperator{\Pa}{Pa}            
\DeclareMathOperator{\Ch}{Ch}            
\DeclareMathOperator{\Anc}{Anc}          
\DeclareMathOperator{\CoPa}{CPa}         

\newtheorem{theorem}{Theorem}
\newtheorem{corollary}[theorem]{Corollary}
\newtheorem{lemma}[theorem]{Lemma}
\theoremstyle{definition}
\newtheorem{definition}{Definition}
\newtheorem{conjecture}{Conjecture}
\newtheorem*{game*}{Causal Game}

\begin{document}
\title{Admissible Causal Structures and Correlations}
\author{Eleftherios-Ermis Tselentis}
	\affiliation{Institute for Quantum Optics and Quantum Information (IQOQI-Vienna), Austrian Academy of Sciences, 1090 Vienna, Austria}
	\affiliation{Faculty of Physics, University of Vienna, 1090 Vienna, Austria}
\author{\"Amin Baumeler}
	\affiliation{Facolt\`a di scienze informatiche, Universit\`a della Svizzera italiana, 6900 Lugano, Switzerland}
	\affiliation{Facolt\`a indipendente di Gandria, 6978 Gandria, Switzerland}

\begin{abstract}
	\noindent
	It is well-known that if one assumes quantum theory to hold locally, then processes with indefinite causal order and cyclic causal structures become feasible. 
	Here, we study {\em qualitative\/} limitations on causal structures and correlations imposed by local quantum theory.
	For one, we find a {\em necessary graph-theoretic criterion\/}---the ``siblings-on-cycles'' property---for a causal structure to be admissible:
	Only such causal structures admit a realization consistent with local quantum theory.
	We conjecture that this property is moreover {\em sufficient.}
	This conjecture is motivated by an explicit construction of quantum causal models, and supported by numerical calculations.
	We show that these causal models, in a restricted setting, are indeed consistent.
	For another, we identify two sets of causal structures that, in the classical-deterministic case, forbid and give rise to non-causal correlations respectively. 
\end{abstract}

\maketitle

\section{Introduction}
At heart of Einstein's equivalence principle is the impossibility to detect the gravitational field via local experiments.\footnote{
	This means that if a non-gravitational experiment is carried out in a sufficiently small space-time region~$\mathcal R$ with a~gravitational field,
	then for any space-time region~$\mathcal R'$ {\em free\/} of gravitation there exists a suitable reference frame where the same experimental procedure yields the identical experimental data.
	This statement and its variations are discussed in Ref.~\cite{dicasola2015}.}
For general relativity, this principle dictates that physics in sufficiently small, {\em i.e., local}, space-time regions is described by special relativity.
This principle naturally extends to the quantum case: {\em local experiments are described by quantum theory.\/}
In this quantum formulation, however, the gravitational field, the reference frames, and the space-time regions might necessitate quantum descriptions.
While different approaches target these descriptions~(see,~{\it e.g.,\/} Refs.~\cite{giacomini2019,causalSets2019,loopQuantumGravity2021,pauloMarios2021}),
another---the process-matrix framework~\cite{oreshkov2012}---abstracts away the general-relativistic freight and focuses on the idealized prescription of {\em local quantum experiments\/} in countably many regions only~(without imposing any global constraints).
Similar to the various formulations of the equivalence principles, this approach can be used to constrain competing theories of quantum gravity:
If a~candidate theory of quantum gravity exceeds the limits of the latter, then local experiments in that theory {\em must\/} disagree with quantum theory.
This, in turn, gives a prescription to experimentally falsify that candidate theory.

The process-matrix framework, {\it i.e.,}~the assumption of local quantum theory, reconciles the inherently {\em probabilistic\/} nature of quantum theory with the {\em dynamical\/} causal structures of general relativity~\cite{hardy2005}.
For one, it extends quantum indefiniteness of physical degrees of freedom like position and momentum to causal connections.
Exemplary, while the position of a mass in general relativity determines the causal order among events in its future, the {\em quantum-switch\/} process~\cite{chiribella2013} does so coherently~\cite{oreshkov2016,zych2017,araujo2015}.
For another, this framework allows for violations of causal inequalities~\cite{oreshkov2012}.
Causal inequalities, similar to Bell inequalities~\cite{bell1964}, are device-independent tests of a global causal order.
If the observed correlations violate such an inequality, then they cannot be causally explained:
Any explanation where only past data influences future observations fails.
These correlations are called {\em non-causal,} and arise in setups that resemble~\cite{baumeler2019} closed time-like curves~(CTCs)~\cite{lanczos1924,godel1949}.
As notoriously shown by G\"odel~\cite{godel1949}, CTCs appear in solutions to Einstein's equation of general relativity.

This stipulation of local quantum theory is also of interest in theoretical computer science.
A pillar of computer sciences is that machines~(programs) and data are treated on an equal footing.
This paradigm finds its climax in Church's notion of computation---the~$\lambda$ calculus~\cite{barendregt1984}---where any data {\em is\/} a~function, and therefore functions are of higher-order: functions on functions.
The process-matrix framework describes the first level of {\em higher-order quantum computation\/}~\cite{perinotti2017,bisio2019}:
Its objects---the process matrices---map quantum gates to quantum gates.
For instance, the previously mentioned quantum switch maps two quantum gates~$A,B$ to the functionality~\mbox{$(\alpha\ket 0+\beta\ket1)\otimes\ket\psi\mapsto\alpha\ket0\otimes BA\ket\psi+\beta\ket1\otimes AB\ket\psi$} where the {\em order\/} of gate application is controlled by the first qubit.
This is achieved, {\it e.g.,\/}~through programmable connections between gates~\cite{colnaghi2012}.
The quantum switch brings forth a reduction in query complexity when compared to the standard circuit model of computation~\mbox{\cite{chiribella2013,araujo2014,renner2021,renner2021a}}.

The causal relations among local quantum experiments~(gates) are conveniently expressed with {\em causal structures.}
A causal structure is a directed graph where the vertices represent laboratories, and where the edges indicate the possibility of a local laboratory to directly influence another (see Figure~\ref{fig:intro}).
The causal relations among the gates of any quantum circuit form an {\em acyclic\/} causal structure:
Naturally, a gate at depth~$d$ of the circuit has no causal influence on the input to any other gate at the same or smaller depth.
This is radically contrasted by processes:
The quantum switch, for instance, has a {\em cyclic\/} causal structure~\cite{barrett2021}.
Still, not {\em every\/} causal structure is compatible with local quantum theory:
For the output of a laboratory to influence the same laboratory's input, we require a departure from quantum theory by introducing non-linear dynamics~\cite{deutsch1991,bennett,lloyd2011}.
\begin{figure}
	\centering
	\begin{subfigure}[c]{0.15\textwidth}
		\centering
		\begin{tikzpicture}
			\def\dd{3}
			\def\ww{.25}
			\tikzstyle{vertex} = [fill,draw,circle,minimum size=\dd,inner sep=0pt]
			\draw (0,0) node[above left] {$A$} -- ++(\ww,0);
			\draw (0,-.5) node[above left] {$B$} -- ++(2*\ww+.5,0);
			\draw (\ww,-.25) rectangle node {$C$} ++(.5,.5);
			\draw (\ww+.5,0) -- ++(\ww,0);
			\draw (2*\ww+.5,-.75) rectangle node {$D$} (2*\ww+1,.25);
			\draw (2*\ww+1,0) -- ++(\ww,0) node[above right] {$E$};
			\draw (2*\ww+1,-.5) -- ++(\ww,0) node[above right] {$F$};
			\def\vdist{-2.25}
			\node[vertex,label=120:{$A$}] (A) at (0,\vdist+.5) {};
			\node[vertex,label=210:{$B$}] (B) at (0,\vdist-.5) {};
			\node[vertex,label=90:{$C$}] (C) at (\ww+.25,\vdist+.25) {};
			\node[vertex,label=270:{$D$}] (D) at (2*\ww+.75,\vdist) {};
			\node[vertex,label=60:{$E$}] (E) at (3*\ww+1,\vdist+.5) {};
			\node[vertex,label=-60:{$F$}] (F) at (3*\ww+1,\vdist-.5) {};
			\draw[-stealth] (A) -- (C);
			\draw[-stealth]	(B) -- (D);
			\draw[-stealth]	(C) -- (D);
			\draw[-stealth]	(D) -- (E);
			\draw[-stealth]	(D) -- (F);
		\end{tikzpicture}
		\caption{}
		\label{fig:introa}
	\end{subfigure}
	\begin{subfigure}[c]{0.15\textwidth}
		\centering
		\begin{tikzpicture}
			\def\dd{3}
			\tikzstyle{vertex} = [fill,draw,circle,minimum size=\dd,inner sep=0pt]
			\def\ww{.25}
			\draw (0,0) node[above left] {$P$} -- ++(\ww,0);
			\draw (\ww,0) |- ++(.25,.75) |- ++(.5,-.65) |- ++(.25,.65) |- ++(-.25,-1.5) |- ++(-.5,.65) |- ++(-.25,-.65) -- cycle;
			\node {} (\ww+.25,.1) rectangle node {$A$} ++(.5,.5);
			\node {} (\ww+.25,-.1) rectangle node {$B$}  ++(.5,-.5);
			\draw (\ww+.25,.1+.25) -- ++(.1,0);
			\draw (\ww+.25+.5,.1+.25) -- ++(-.1,0);
			\draw (\ww+.25,-.1-.25) -- ++(.1,0);
			\draw (\ww+.25+.5,-.1-.25) -- ++(-.1,0);
			\draw (\ww+1,0) -- ++(\ww,0) node[above right] {$F$};
			\def\vdist{-2}
			\node[vertex,label=120:{$P$}] (P) at (0,\vdist) {};
			\node[vertex,label=90:{$A$}] (A) at ($ (P.center) + (.75,.5) $) {};
			\node[vertex,label=-90:{$B$}] (B) at ($ (P.center) + (.75,-.5) $) {};
			\node[vertex,label=60:{$F$}] (F) at ($ (P.center) + (1.5,0) $) {};
			\draw[-stealth] (P) -- (A);
			\draw[-stealth] (P) -- (B);
			\draw[-stealth] (P) -- (F);
			\draw[-stealth] (A) -- (F);
			\draw[-stealth] (B) -- (F);
			\draw[-stealth] (A) to[out=270+30,in=90-30] (B);
			\draw[-stealth] (B) to[out=90+30,in=270-30] (A);
		\end{tikzpicture}
		\caption{}
		\label{fig:introb}
	\end{subfigure}
	\begin{subfigure}[c]{0.1\textwidth}
		\centering
		\begin{tikzpicture}
			\def\dd{3}
			\tikzstyle{vertex} = [fill,draw,circle,minimum size=\dd,inner sep=0pt]
			\draw (0,0) rectangle node {$A$} ++(.5,1);
			\def\ww{.25}
			\draw[rounded corners,->] (.5,.3) -- ++(\ww,0) |- (-\ww,-.5) |- (0,.3);
			\draw[<-] (0,.7) -- ++(-\ww,0) node[above left] {$P$};
			\draw[->] (.5,.7) -- ++(\ww,0) node[above right] {$F$};
			\def\vdist{-2.1}
			\node[vertex,label=270:{$A$}] (A) at (0,\vdist) {};
			\node[vertex,label=120:{$P$}] (P) at ($ (A.center) - (\ww+.2,0) $) {};
			\node[vertex,label=60:{$F$}] (F) at ($ (A.center) + (\ww+.2,0) $) {};
			\draw[-stealth] (A) to[out=90-30,in=90+30,looseness=40] (A);
			\draw[-stealth] (P) -- (A);
			\draw[-stealth] (A) -- (F);
		\end{tikzpicture}
		\caption{}
		\label{fig:introc}
	\end{subfigure}
	\caption{(a) A quantum circuit and its {\em acyclic\/} causal structure. (b) The quantum switch---an instance of the process-matrix framework---has a {\em cyclic\/} causal structure: Depending on the prepared state at~$P$, a quantum system is sent through the H-shaped region from~$A$ to~$B$ or from~$B$ to~$A$. (c) If~$A$ is traversed by a closed time-like curve, then~$A$'s output influences the input, and a departure from quantum theory becomes necessary.}
	\label{fig:intro}
\end{figure}
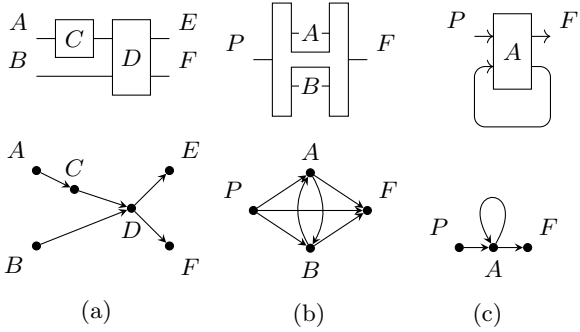

In this work, we study the causal structures that admit a quantum realization---a question raised in Ref.~\cite{baumeler2021}.
In other words: We study the possible causal relations among laboratories under the assumptions that within each laboratory no deviation from quantum theory is observable.
We find a {\em necessary graph-theoretic criterion\/} (the causal structure of every quantum process satisfies this criterion) and conjecture that the criterion is also {\em sufficient.}
The conjecture is motivated by a construction of causal models for the causal structures of interest, and is moreover numerically tested for all directed graph with up to six nodes.
In addition, we provide two graph-theoretic criteria from which, in the classical-deterministic case, only causal or also non-causal correlations arise.
Supporting the above conjecture, we show that the causal structures that satisfy the criterion for non-violation are admissible.

The presentation is structured in the following way.
First, we provide the mathematical tools necessary for the present treatment.
This is followed by our results on admissible and inadmissible causal structures. Thereafter, we relate causal structures with causal inequalities.
We conclude with a series of open questions.

\section{Preliminaries}
We briefly comment on the notation used.
If~$\mathcal H$ is a~Hilbert space, then~$\mathcal L(\mathcal H)$ is the set of linear operators on~$\mathcal H$.
We use~$\mathbb Z_n$ for the set~$\{0,1,\dots,n-1\}$.
If a~symbol is used with and without a subscript from~$\mathbb Z_n$, then the bare symbol denotes the collection under the natural composition, {\it e.g.,\/}~we will use~$x$ to denote~$(x_k)_{k\in\mathbb Z_n}$.
If the subscript is a subset of~$\mathbb Z_n$, then the composition is taken only over those elements.
Moreover, we use~$\setminus S$ as shorthand for~$\mathbb Z_n\setminus S$, and~$\setminus i$ for~$\setminus \{i\}$.
If~$\varepsilon$ is a completely positive map, then~$\rho^\varepsilon$ is its Choi operator~\cite{choi1975}.
For a directed graph~\mbox{$G=(V,E\subseteq V\times V)$},~$V$ denotes the set of nodes and~$E$ the set of directed edges.
A~{\em directed path\/}~\mbox{$\pi=(v_0,\dots,v_\ell)$} is a sequence of distinct nodes with~\mbox{$\{(v_i,v_{i+1})\mid 0\leq i < \ell\}\subseteq E$}.
A~{\em directed cycle\/}~\mbox{$C=(v_0,\dots,v_\ell)$} is a directed path with~$(v_\ell,v_0)\in E$.
The induced graph $G[V']$ is the graph $G'=(V',E')$ where~$V'\subseteq V$ and all edges~$E'\subseteq E$ have endpoints in~$V'$, {\em i.e.,}~$(i,j)\in E'$ if and only if~\mbox{$i,j\in V'$} and~$(i,j)\in E$.
A~directed cycle~$C$ is called {\em induced directed cycle\/} or {\em chordless cycle\/} if the induced graph~$G[C]$ is a directed cycle graph.
For a~node~$k\in V$, we use~$\Pa(k)$,~$\Ch(k)$, and~$\Anc(k)$ to express the set of parents, children, and ancestors of~$k$, respectively, and similarly~$\Pa(S)$,~$\Ch(S)$ and $\Anc(S)$ by taking the union over a set~$S$.
The cardinality of the set $\Pa(k)$ is the {\em in-degree\/} $\deg_\text{in}(k)$.
A~node with zero in-degree is called {\em source.}
We use $\CoPa(S)$ for the union of the common parents of all elements $i\neq j\in S$, {\em i.e.,} $\CoPa(S):=\bigcup_{i\neq j\in S}\Pa(i)\cap \Pa(j)$.
Two nodes~$i,j\in V$ are called {\em siblings\/} if and only if they have common parents, {\it i.e.,\/}~$\CoPa(\{i,j\})\neq\emptyset$.
A directed path~$\pi$ is said to {\em contain siblings\/} if and only if~$\CoPa(\pi)\neq\emptyset$.

\subsection{Correlations}
Correlations observed among~$n$ parties (regions)~$\mathbb Z_n$ are expressed with the conditional probability distribution~$p(a\mid x)$
where for each party~$k$ we have~$a_k\in \mathcal A_k$, and~$x_k\in\mathcal X_k$.
The set~$\mathcal X_k$ is the set of experimental {\em settings\,} and~$\mathcal A_k$ the set of experimental {\em observations.}
\begin{definition}[Causal correlations~\cite{oreshkov2012, oreshkov2016, alastaircorr2016}]
	\label{def:causal}
	The~$n$-party correlations~$p(a\mid x$) are {\em causal\/} if and only if they can be decomposed as
	\begin{align}
		p(a\mid x) = \sum_{k\in\mathbb Z_n} q_k p(a_k\mid x_k) p_{a_k}^{x_k}(a_{\setminus k}\mid x_{\setminus k})
		\,,
	\end{align}
	where~$\forall k\in\mathbb Z_n: q_k \geq 0$,~$\sum_{k\in\mathbb Z_n} q_k=1$, and~\mbox{$p_{a_k}^{x_k}(a_{\setminus k}\mid x_{\setminus k})$} are~$(n-1)$-party causal correlations.
	If this decomposition is infeasible, the correlations are called {\em non-causal.}
\end{definition}
The motivation behind this definition is that each party can influence its future only, including the causal order of the parties in its future. 
From this follows that there exists at least one party~$k$ whose observation does not depend on the data of any other party (the value of~$a_k$ depends solely on~$x_k$).
Moreover, the causal order among the parties might be subject to randomness, {\it e.g.,\/}~a coin flip.

\subsection{Processes}
In the process-matrix framework~\cite{oreshkov2012}, each party (region) is defined through a past and future space-like boundary~(see Figure~\ref{fig:process}).
\begin{figure}
	\centering
	\begin{tikzpicture}
		\def\dv{0.75cm}
		\coordinate (A) at (0,0);
		\coordinate (AL) at ($ (A) - (\dv,0) $);
		\coordinate (AR) at ($ (A) + (\dv,0) $);
		\coordinate (B) at (2,0);
		\coordinate (BL) at ($ (B) - (\dv,0) $);
		\coordinate (BR) at ($ (B) + (\dv,0) $);
		\coordinate (C) at (4,0);
		\coordinate (CL) at ($ (C) - (\dv,0) $);
		\coordinate (CR) at ($ (C) + (\dv,0) $);
		\def\alpha{60}
		\path[-,thick] (AL) to[out=\alpha,in=180-\alpha] node[midway] (AT) {} (AR) to[out=\alpha-180,in=-\alpha] node[midway] (AB) {} (AL);
		\path[-,thick] (BL) to[out=\alpha,in=180-\alpha] node[midway] (BT) {} (BR) to[out=\alpha-180,in=-\alpha] node[midway] (BB) {} (BL);
		\path[-,thick] (CL) to[out=\alpha,in=180-\alpha] node[midway] (CT) {} (CR) to[out=\alpha-180,in=-\alpha] node[midway] (CB) {} (CL);
		\def\d{.6}
		\def\vv{1.25}
		\def\hh{.3}
		\draw[fill=lightgray] ($ (AT) + (-\vv-\d,\hh+\d) $) -- ($ (CT) + (\vv+\d,\hh+\d) $)
		-- ($ (CB) + (\vv+\d,-\hh-\d) $) -- ($ (AB) + (-\vv-\d,-\hh-\d) $)
		-- cycle;
		\draw[fill=white] ($ (AT) + (-\vv,\hh) $) -- ($ (CT) + (\vv,\hh) $)
		-- ($ (CB) + (\vv,-\hh) $) -- ($ (AB) + (-\vv,-\hh) $)
		-- cycle;
		\draw[-,thick] (AL) to[out=\alpha,in=180-\alpha] node[pos=.2,above] {$\mathcal O_0$} node[midway] (AT) {} (AR) to[out=\alpha-180,in=-\alpha] node[pos=.8,below] {$\mathcal I_0$} node[midway] (AB) {} (AL);
		\draw[-,thick] (BL) to[out=\alpha,in=180-\alpha] node[pos=.2,above] {$\mathcal O_1$} node[midway] (BT) {} (BR) to[out=\alpha-180,in=-\alpha] node[pos=.8,below] {$\mathcal I_1$} node[midway] (BB) {} (BL);
		\draw[-,thick] (CL) to[out=\alpha,in=180-\alpha] node[pos=.2,above] {$\mathcal O_2$} node[midway] (CT) {} (CR) to[out=\alpha-180,in=-\alpha] node[pos=.8,below] {$\mathcal I_2$} node[midway] (CB) {} (CL);
		\node (Acc) at (A) {$\mu_0^{a_0|x_0}$};
		\node (Bcc) at (B) {$\mu_1^{a_1|x_1}$};
		\node (Ccc) at (C) {$\mu_2^{a_2|x_2}$};
		\draw[>=stealth,double,->] (AT.center) -- ++(0,\hh);
		\draw[>=stealth,double,->] (BT.center) -- ++(0,\hh);
		\draw[>=stealth,double,->] (CT.center) -- ++(0,\hh);
		\draw[>=stealth,double,<-] (AB.center) -- ++(0,-\hh);
		\draw[>=stealth,double,<-] (BB.center) -- ++(0,-\hh);
		\draw[>=stealth,double,<-] (CB.center) -- ++(0,-\hh);
		\draw[rounded corners,thick,red,->] ($ (AT) + (0,\hh) $)
		-- ($ (AT) + (0,\hh+\d/2) $)
		-- ($ (AT) + (-\vv-\d/2,\hh+\d/2) $)
		-- ($ (AB) + (-\vv-\d/2,-\hh-\d/2) $)
		-- ($ (BB) + (0,-\hh-\d/2) $)
		-- ($ (BB) + (0,-\hh) $);
		\draw[rounded corners,thick,red,->] ($ (BT) + (0,\hh) $)
		-- ($ (BT) + (0,\hh+\d/2) $)
		-- ($ (CT) + (\vv+\d/2,\hh+\d/2) $)
		-- ($ (CB) + (\vv+\d/2,-\hh-\d/2) $)
		-- ($ (CB) + (0,-\hh-\d/2) $)
		-- ($ (CB) + (0,-\hh) $);
	\end{tikzpicture}
	\caption{Schematic of three parties and a process. The gray area represents the process:
	It takes the systems on the future boundaries of the parties and maps it to the past boundaries of the parties.
	The red connections indicate an example where the parties are causally ordered increasingly.
	{\it A priori,} however, the process is not assumed to respect any ordering of the parties.
}
	\label{fig:process}
\end{figure}
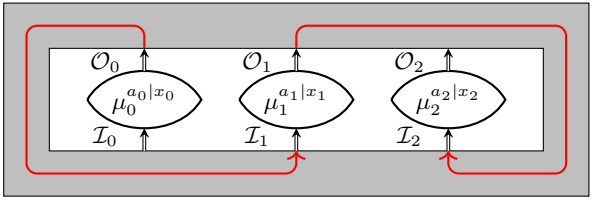
For party~$k$, we denote by~$\mathcal I_k$ the Hilbert space on the past, by~$\mathcal O_k$ the Hilbert space on the future boundary, and by~$\cpm_k$ ($\cptp_k$) the set of all completely positive~(trace-preserving) maps from~$\mathcal L(\mathcal I_k)$ to~$\mathcal L(\mathcal O_k)$.
A~quantum experiment for party~$k$ is a quantum channel from~$\mathcal I_k$ to~$\mathcal O_k$, equipped with a classical input (the setting) and a~classical output (the observation).
Hence, it is a~family~\mbox{$\{\mu^{a_k\mid x_k}_k\in\cpm_k\}_{(a_k,x_k)\in\mathcal A_k\times\mathcal X_k}$} of maps, such that for all~$x_k\in\mathcal X_k$, we have~\mbox{$\sum_{a_k\in\mathcal A_k}\mu_k^{a_k\mid x_k}\in\cptp_k$}.
A process interlinks all parties without any assumption on their causal relations, but with the sole assumption that no deviation from quantum theory is locally observable, {\it i.e.,\/}~the probability distribution over the observations is a multi-linear function of the quantum experiments and well-defined (even if the parties share an arbitrary quantum system).
\begin{definition}[Process~\cite{oreshkov2012}]
	An~$n$-party {\em quantum process\/} is a positive semi-definite operator~\mbox{$W\in\mathcal L(\mathcal I\otimes\mathcal O)$}
	with
	\begin{align}
		\forall \{\mu_k\in\cptp_k\}_{k\in\mathbb Z_n}:
		\Tr\left[W\bigotimes_{k\in\mathbb Z_n}\rho^{\mu_k}\right] = 1
		\,.
		\label{eq:qprocess}
	\end{align}
\end{definition}
Note that in this definition, the experimental settings and observations are absent, or equivalently, the sets~$\mathcal A_k, \mathcal X_k$ are singletons.
Eq.~\eqref{eq:qprocess} states that the total probability of observing this single outcome under this single setting is one;
this holds for any setting and is independent of any resolution of the completely positive trace-preserving map into maps with classical outputs.
Oreshkov, Costa, and Brukner~\cite{oreshkov2012} already observed that the quantum process~$W$ is the Choi operator of a~completely positive trace-preserving map from all future boundaries to all past boundaries of the parties ({\it cf.}~Figure~\ref{fig:process}).
For a~given choice of~$n$-party quantum process and experiments, the correlations are computed with the generalized Born rule
\begin{align}
	p(a\mid x)
	:=
	\Tr\left[
		W
		\bigotimes_{k\in\mathbb Z_n}
		\rho^{\mu_k^{a_k\mid x_k}}
	\right]
	\,.
\end{align}

If one assumes that the parties perform classical-deterministic experiments, as opposed to quantum experiments, we arrive at the following special case.
The spaces on the past and future boundaries~$\mathcal I_k,\mathcal O_k$ are sets (as opposed to Hilbert spaces),
and an experiment for party~$k$ is a
function~\mbox{$\mu_k:\mathcal I_k\times\mathcal X_k\rightarrow\mathcal O_k\times\mathcal A_k$} (as opposed to a family of maps),
where we use~\mbox{$\mu_k^0:\mathcal I_k\times\mathcal X_k\rightarrow\mathcal O_k$} for the first, and~\mbox{$\mu_k^1:\mathcal I_k\times\mathcal X_k\rightarrow\mathcal A_k$} for the second component.
Just as in the quantum case, a classical-deterministic process~$\omega$ turns out to be a function from the future to the past boundaries, and Eq.~\eqref{eq:qprocess} translates into an intuitive condition:
For any choice of experiment, there exists a~{\em unique\/} consistent assignment of values to the input spaces; the map~$\omega\circ\mu$ has a~unique fixed point.
\begin{theorem}[Classical-determintic process~\cite{baumeler2016}]
	An~$n$-party {\em classical-deterministic process\/} is a function~\mbox{$\omega:\mathcal O\rightarrow\mathcal I$}
	with
	\begin{align}
		\forall \{\mu_k:\mathcal I_k\rightarrow \mathcal O_k\}_{k\in\mathbb Z_n},\,\exists ! (r_k)_{k\in\mathbb Z_n}:
		r = \omega(\mu(r))
		\,,
		\label{eq:cprocess}
	\end{align}
	where~$\exists!$ is the uniqueness quantifier.
	\label{thm:cprocess}
\end{theorem}
Here, the correlations among the~$n$-parties are computed via
\begin{align}
	p(a|x): = \omega \star \mu^{a\mid x}
	\!=\sum_{i,o}
	\left[\omega (o)=i\right]
	\left[(o,a)=\mu(x,i)\right]
	\,,
	\label{eq:correlations}
\end{align}
where we use~$[n=m]$ for the Kronecker delta~$\delta_{n,m}$, and~$\star$ for the link product~\cite{chiribella2008}.
Note that every~$n$-party classical-deterministic process~$\omega$ also corresponds to an~$n$-party quantum process~\cite{baumeler2021}
\begin{align}
	W_\omega := \sum_{o\in\mathcal O} \ket{o}\bra{o}_{\mathcal O} \otimes \ket{\omega(o)}\bra{\omega(o)}_{\mathcal I}
	\,,
	\label{eq:classical2quantum}
\end{align}
where~$\ket o=\bigotimes_{k\in\mathbb Z_n}\ket {o_k}$, with~$\{\ket {o'_k}\}_{o'_k\in\mathcal O_k}$ being a basis for all~$k\in\mathbb Z_n$, and similarly for~$\ket{\omega(o)}$.

We briefly illustrate the above theorem in the single-party case~$(n=1)$.
Assume the function~$\omega$ is the identity function from~$\mathcal O=\{0,1\}$ to~$\mathcal I=\{0,1\}$.
Here,~$\omega$ describes a {\em closed time-like curve:} The output of that single party is identically mapped to the same party's input.
Suppose now that this single party implements the negation~\mbox{$\mu(b)=1-b$}.
What is the state of the system on the party's past boundary?
In this instance of time-travel antinomy---the grandparent antinomy---no consistent specification is possible:
If it is~$b$, then the local experiment specifies the state on the future boundary to be~$1-b$, and in turn, the function~$\omega$ specifies the state on the past boundary to be~$1-b$; but it is~$b$.
In other terms, the function~$\omega\circ\mu$ has {\em no\/} fixed point; this function~$\omega$ is {\em not\/} a classical-deterministic process.
Another instance of a~time-travel antinomy---the information antinomy---arises if the party implements the {\em identity\/} experiment~$\mu(b)=b$.
Then, the state on the party's past boundary is not determined: both values,~$0$ and~$1$, are equally justifiable; the function~$\omega\circ\mu$ has {\em two\/} fixed points.
As it turns out, these two antinomies are equivalent:
\begin{theorem}[Equivalence of antinomies~\cite{baumeler2020equivalence}]
	The~$n$-party function~$\omega:\mathcal O\rightarrow\mathcal I$ suffers from the grandparent antinomy, {\it i.e.,}~there exists a choice of experiments~\mbox{$\{\mu_k:\mathcal I_k\rightarrow\mathcal O_k\}_{k\in\mathbb Z_n}$} such that~$\omega\circ\mu$ has {\em no\/} fixed points,
	if and only if~$\omega$ suffers from the information antinomy, {\it i.e.,}~there exists a~$\mu$ such that~$\omega\circ\mu$ has {\em two or more\/} fixed points.
	\label{thm:equivalence}
\end{theorem}
In the single-party case, the only classical-deterministic processes are the constant functions~$\omega(b)=c$, the {\em unique\/} fixed point being~$c$.
For three or more parties, classical-deterministic processes exist that allow for non-causal correlations:
Any simulation of these correlations requires the abandonment of a causal order, resulting in a form of closed time-like curves~\cite{baumeler2019}, or, as has been recently suggested, time-delocalized systems~\cite{wechs2022}.

In the present treatment, we will also make use of {\em reduced processes:}
If we invoke an experiment~$\mu_k$ of a~single party~$k$, but leave all other experiments unspecified, then the reduced function is a process again.
To do so, we specify the state on the future boundary of party~$k$, which is~$\mu_k(i_k)$, where~$i_k$ is the state on the past boundary of party~$k$.
The state~$i_k$ is well-defined because the process is component-wise non-signaling:~$i_k:=\omega_k(o)$ is independent of~$o_k$.
\begin{lemma}[Component-wise non-signaling and reduced process~\cite{baumeler2019,baumeler2020equivalence}]
	If~$\omega:\mathcal O\rightarrow\mathcal I$ is an~$n$-party classical-deterministic process, then it is component-wise non-signaling, {\it i.e.,}
	the component~$\omega_k$ does not depend on the~$k$-th input:
	\begin{align}
		\forall k\in\mathbb Z_n,\forall o\in\mathcal O,o'_k\in\mathcal O_k:
		\omega_k\left(o\right)
		=
		\omega_k\left(o'_k,o_{\setminus k}\right)
		\,.
	\end{align}
	If~$\omega:\mathcal O\rightarrow\mathcal I$ is an~$(n\geq2)$-party classical-deterministic process, then for all~$k\in\mathbb Z_n$ and~\mbox{$\mu_k:\mathcal I_k\rightarrow\mathcal O_k$}, the {\em reduced function\/}~\mbox{$\omega^{\mu_k}_{\setminus k}:\mathcal O_{\setminus k} \rightarrow \mathcal I_{\setminus k}$}
	defined for all~\mbox{$\ell\in\mathbb Z_n\setminus \{k\}$} via
	\begin{align}
		\omega^{\mu_k}_{\ell}: \mathcal O_{\setminus k} &\rightarrow\mathcal I_{\ell}\\
		o_{\setminus k} &\mapsto \omega_\ell\left( o_{\setminus k}, \mu_k\left( \omega_k\left( o_{\setminus k} \right) \right) \right)
	\end{align}
	is an~$(n-1)$-party classical-deterministic process.
	\label{lemma:reduced}
\end{lemma}

The process establishes the causal connections among the parties.
Such connections---the causal structure---are conveniently expressed with a directed graph where the nodes~$\mathbb Z_n$ represent the parties, and where the {\em absence\/} of an edge from~$i$ to~$j$ indicates that the process is not signaling from the future boundary of party~$i$ to the past boundary of party~$j$.
The causal structure of the process schematically represented with the red arrows in Figure~\ref{fig:process} has no edge from~$1$ to~$0$, from~$2$ to~$0$ and from~$2$ to~$1$.

\subsection{Causal models}
A causal model is a causal structure (a directed graph) equipped with model parameters (channels along the edges).
Traditionally, the nodes of a causal model are random variables~\cite{pearl2009}.
Here, in contrast, we adopt the split-node model~\cite{richardson2013}, where each node is split into an incoming and an outgoing part, the past and future boundary ({\it cf.}~Figure~\ref{fig:QCMexample}).
We follow the recent work by Barrett, Lorenz, and Oreshkov~\cite{barrett2021} which unifies processes and causal models.
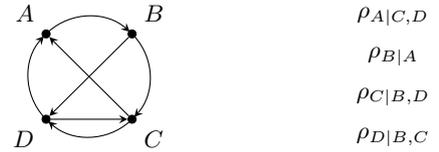
\begin{figure}
	\centering
	\begin{tikzpicture}
		\def\radius{.8}
		\def\rotate{45}
		\def\anga{90+\rotate}
		\def\angb{0+\rotate}
		\def\angc{-90+\rotate}
		\def\angd{180+\rotate}
		\def\dd{3}
		\tikzstyle{vertex} = [fill,draw,circle,minimum size=\dd,inner sep=0pt]
		\node[vertex,label=\anga:{$A$}] (A) at (\anga:\radius) {};
		\node[vertex,label=\angb:{$B$}] (B) at (\angb:\radius) {};
		\node[vertex,label=\angc:{$C$}] (C) at (\angc:\radius) {};
		\node[vertex,label=\angd:{$D$}] (D) at (\angd:\radius) {};
		\draw[-stealth] (\anga-\dd:\radius) arc(\anga-\dd:\angb+\dd:\radius);
		\draw[-stealth] (\angb-\dd:\radius) arc(\angb-\dd:\angc+\dd:\radius);
		\draw[-stealth] (\angc-\dd:\radius) arc(360+\angc-\dd:\angd+\dd:\radius);
		\draw[-stealth] (\angd-\dd:\radius) arc(\angd-\dd:\anga+\dd:\radius);
		\draw[-stealth] (D) -- (C);
		\draw[-stealth] (C) -- (A);
		\draw[-stealth] (B) -- (D);
		\node (MA) at (4,\radius) {$\rho_{A\mid C,D}$};
		\node (MB) at (4,\radius/3) {$\rho_{B\mid A}$};
		\node (MC) at (4,-\radius/3) {$\rho_{C\mid B,D}$};
		\node (MD) at (4,-\radius) {$\rho_{D\mid B,C}$};
	\end{tikzpicture}
	\caption{Example of a four-node causal model. The state on the input space of node~$A$ is obtained by evolving the state on the output space of nodes~$C,D$ through the channel~$\rho_{A\mid C,D}$, and similarly for the other nodes.}
	\label{fig:QCMexample}
\end{figure}
\begin{definition}[Causal model, consistency, and faithfulness~\cite{barrett2021}]
	An~$n$-party {\em causal model\/} is a directed graph~$G=(\mathbb Z_n,E)$ (causal structure) equipped with~$\{\rho_{k\mid\Pa(k)}\}_{k\in\mathbb Z_n}$ (model parameters).
	In the {\em classical-deterministic\/} case, the model parameters are functions~$\mathcal O_{\Pa(k)}\rightarrow\mathcal I_k$, and define a {\em classical map\/}~\mbox{$\omega:=(\rho_{k\mid\Pa(k)})_{k\in\mathbb Z_n}$}.
	In the {\em quantum\/} case, the model parameters are the Choi operators of completely positive trace-preserving maps~$\mathcal L(\mathcal O_{\Pa(k)})\rightarrow\mathcal L(\mathcal I_k)$, such that~$\forall i,j\in\mathbb Z_n:[\rho_{i\mid\Pa(i)},\rho_{j\mid\Pa(j)}]=0$, and define a {\em quantum map\/}~$W:=\prod_{k\in\mathbb Z_n}\rho_{k\mid\Pa(k)}$.
	The causal model is {\em consistent\/} if and only if~$\omega,W$, is an~$n$-party classical-deterministic or quantum process, respectively.
	The causal model is {\em faithful\/} if and only if for all~$k$, the model parameter~$\rho_{k\mid\Pa(k)}$ is signaling from all~$\ell\in\Pa(k)$ to~$k$.
\end{definition}
Allow us to comment on some aspects of the above definition.
First---and as has been commented on earlier---if~$\omega$ is a classical-deterministic process, then~$W_\omega$ (see Equation~\eqref{eq:classical2quantum}) is a quantum process.
So, the above definition could be adjusted to refer to quantum processes only.
Yet, while in the following we will make extensive use of the properties of classical-deterministic processes, we explicitly define classical-deterministic causal models in reference to classical-deterministic processes.
Second, the commutativity criterion above ensures Markovianity of the quantum process~\cite{barrett2021}.
Naturally, only Markov processes can be faithfully represented as a~causal model;
without Markovianity, the representation of a causal structure as a~directed graph would be meaningless.
Note that not every process admits a description as a~faithful causal model.
This is the case for the initial two-party process~$W^\text{OCB}$~\cite{oreshkov2012}, which does not admit a factorization into commuting Choi operators~$\rho_{\text{Alice}\mid\text{Bob}},\rho_{\text{Bob}\mid\text{Alice}}$.
Simultaneously,~$W^\text{OCB}$ is not unitarily extensible~\cite{arajo2017}.
It is conjectured that these conditions are equivalent~\cite{barrett2021}.

\section{Admissible causal structures}
The definition of causal models allows us to precise the notion of admissible causal structures:
A causal structure {\em admits\/} a quantum realization whenever it can be amended with model parameters to obtain a {\em faithful and consistent\/} causal model.
An {\em inadmissible\/} causal structure is therefore {\em incompatible\/} with local quantum theory.
The requirement of faithfulness ensures that all non-signaling {\em and signaling\/} relations are expressed by the causal structure.
If we were to neglect faithfulness, then {\em every\/} directed graph is trivially admissible:
The constant~(state-preparation) model parameters that provide the qubit~$\ket 0$ to each party~$k$ satisfy {\em any\/} non-signaling requirement.
A relevant graph-theoretic criterion turns out to be the siblings-on-cycles property of a graph~$G$.
\begin{definition}[Siblings-on-cycles graph]
	A directed graph~$G=(V,E)$ is a {\em siblings-on-cycles\/} graph if and only if each directed cycle in~$G$ contains siblings.
	\label{def:soc}
\end{definition}

For our first result---the characterization of {\em inadmissible\/} causal structures---we make use of the following lemma, which ensures that signals can be propagated along sibling-free paths.
Intuitively, the influences on the parties along the path come from the ``previous'' party on the path and ``outside'' parties.
Faithfulness then guarantees that the signal can be sent along the path.
In contrast, if there were siblings, the common parent might block the signal propagation along the path.
\begin{lemma}[Quantum signaling path]
	\label{lemma:signaling}
	Consider a faithful quantum causal model with causal structure~$G=(V,E)$ and model parameters~$\{\rho_{k\mid\Pa(k)}\}_{k\in V}$.
	If~$\pi=(v,\dots,w)$ is a directed path in~$G$ {\em without siblings,}
	then there exist local experiments such that party~$v$ can signal to party~$w$.
\end{lemma}
\begin{proof}
	As the nodes in~$\pi$ do not have common parents, the quantum map~$W$ is~$\bigotimes_{k\in \pi}\rho_{k\mid\Pa(k)}\prod_{k\not\in \pi}\rho_{k\mid\Pa(k)}$.
	We start by partially fixing the local quantum experiment of each party~$k\not\in \pi$ to discard the input on~$\mathcal I_k$, {\it i.e.,\/}~each such party applies the map~$\rho^{\mu'_k}=\id_{\mathcal I_k}$.
	Because~$\rho_{k\mid\Pa(k)}$ is the Choi operator of a completely positive trace-preserving map~$\mathcal L(\mathcal O_{\Pa(k)})\rightarrow\mathcal L(\mathcal I_k)$,
	we have~$\Tr_{\mathcal I_k}[\rho_{k\mid\Pa(k)}]=\id_{\mathcal O_{\Pa(k)}}$, and the resulting reduced quantum map~$\Tr_{\mathcal I_{V\setminus \pi}}W$ is~$\bigotimes_{k\in \pi} \rho_{k\mid\Pa(k)} \id_{\mathcal O_{\Pa(V\setminus \pi)}}$.
	Note that the causal model we started with is faithful:
	For each party~$k$ and each parent~$\ell\in\Pa(k)$, the map~$\rho_{k\mid\Pa(k)}$ is signaling from~$\ell$ to~$k$.
	In other words, there exists a~quantum state~\mbox{$\tau\in\mathcal L(\mathcal O_{\Pa(k)\setminus\{\ell\}})$} such that~\mbox{$\rho'_{k\mid\ell} := \Tr_{\mathcal O_{\Pa(k)\setminus\{\ell\}}}[ \rho_{k\mid\Pa(k)}\tau^T]$}
	is the Choi operator of a {\em signaling\/} completely positive trace-preserving map from~$\mathcal L(\mathcal O_\ell)$ to~$\mathcal L(\mathcal I_k)$.
	Because of this and because the nodes in~$\pi$ do not have common parents, we can complete the local experiments of all parties~$k\in\Pa(\pi)\setminus \pi$ to prepare a state~$\tau\in\mathcal L(\mathcal O_{\Pa(\pi)\setminus \pi})$ where a signal can be sent from one node to the next along~$\pi$, {\it i.e.,\/} the reduced map is signaling from~$\ell$ to~$k$ for all~$k\in\pi,\ell\in\Pa(k)\cap\pi$.
	Finally, a signaling channel from~$\mathcal L(\mathcal O_v)$ to~$\mathcal L(\mathcal I_w)$ is obtained by implementing an appropriate identity channel for each party~$k\in \pi\setminus \{v,w\}$, and by preparing an arbitrary quantum state for all remaining parties.
\end{proof}

Now, we state and prove our first result~(see~Figure~\ref{fig:admissiblethreepartystructures}):
\begin{figure}
	\centering
	\includegraphics[width=0.4\textwidth]{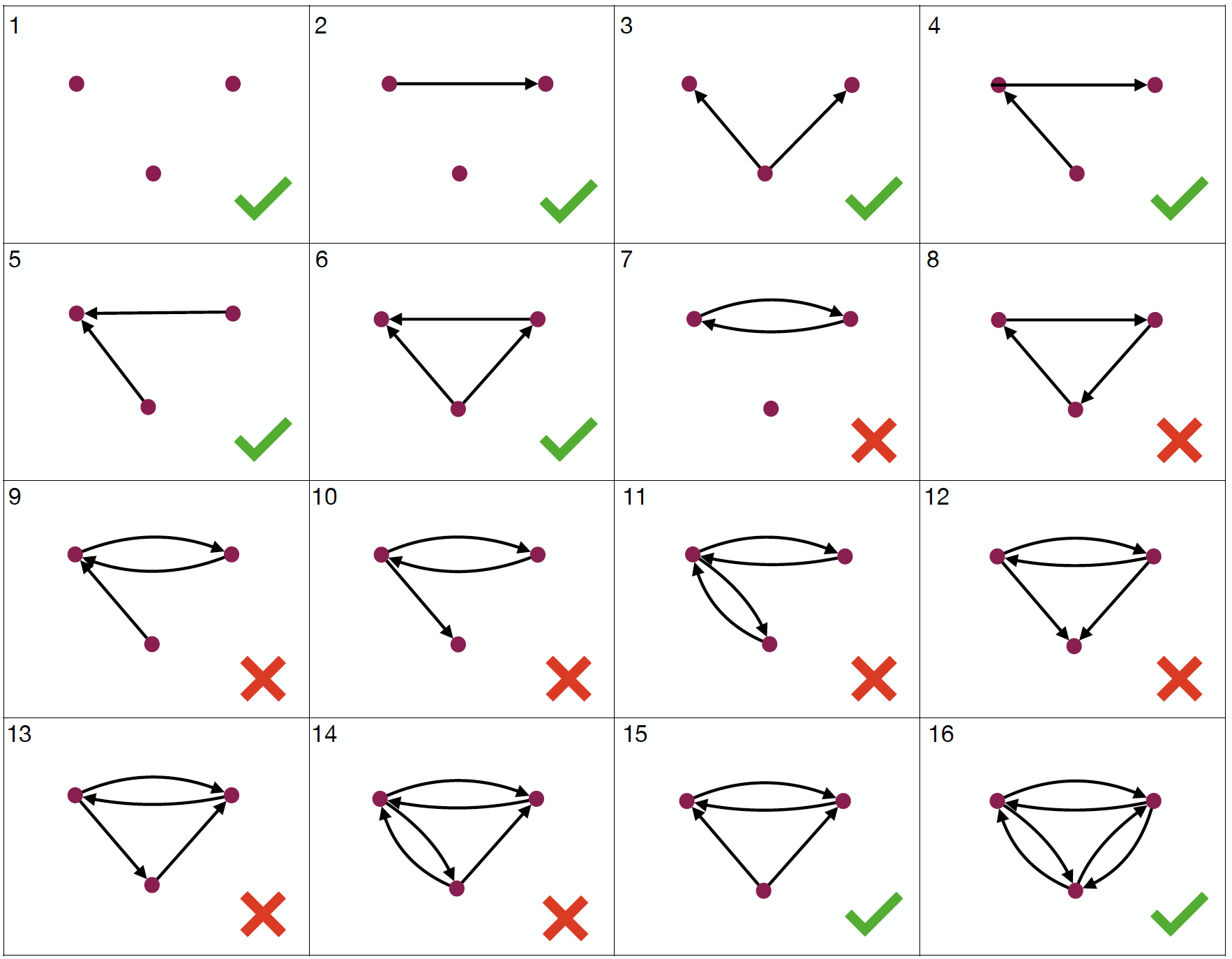}
	\caption{Characterization of all pairwise non-isomorphic causal structures for three parties.
		Graphs 7-14 are inadmissible (Theorem~\ref{thm:inadmissible}).
		In the classical-deterministic case, graph 16 leads to non-causal correlations (Theorem~\ref{thm:noncausal}), and the others to causal correlations only (Theorem~\ref{thm:causal}).
		Graph~16 is also the causal structure of AF/BW process~\cite{af,baumeler2016space}.
		Graph~15 is the causal structure of the quantum switch without a~region in the global future ({\it cf.}~Figure~\ref{fig:introb}).
	}
	\label{fig:admissiblethreepartystructures}
\end{figure}
\begin{theorem}[Inadmissible causal structures]
	\label{thm:inadmissible}
	The causal structure of every faithful and consistent quantum causal model is a siblings-on-cycles graph, or, equivalently,
	if a~graph~$G$ is {\em not\/} a siblings-on-cycles graph, then the causal structure~$G$ is {\em inadmissible.}
\end{theorem}
\begin{proof}
	We prove the latter formulation of the theorem.
	Assume that~$G$ contains a directed cycle~\mbox{$C=(v,\dots,w)$} {\em without siblings,} and let~$(G,\{\rho_{k\mid\Pa(k)}\}_{k\in V})$ be an arbitrary faithful quantum causal model with causal structure~$G$.
	By Lemma~\ref{lemma:signaling}, there exist local quantum experiments for all parties except~$v$ such that the reduced process~$W':=\rho_{v\mid v}$ is the Choi operator of a {\em signaling\/} completely positive trace-preserving map from~$\mathcal L(\mathcal O_v)$ to~$\mathcal L(\mathcal I_v)$.
	This, however, is not a single-party process~\cite{oreshkov2012}: There exists a~map~\mbox{$\mu_v\in\cptp_v$} such that~$\Tr[W'\rho^{\mu_v}]\neq 1$.
\end{proof}

For our second result concerning the admissibility of causal structures, we give a construction of model parameters for any causal structure.
We conjecture that these model parameters {\em always\/} give rise to {\em consistent\/} causal models when the relevant causal structure is a~siblings-on-cycles graph.

\begin{definition}[Model parameters]
	\label{def:modelparameters}
	Let $G=(V,E)$ be a~directed graph.
	For each party~$k\in V$, define the input space~$\mathcal I_{k}:=\mathbb Z_2$, the output space~$\mathcal O_{k}:=\Ch(k)\cup\{\bot\}$, and the model parameters~$\{\rho_{k\mid\Pa(k)}:\mathcal O_{\Pa(k)}\rightarrow\mathcal I_k\}_{k\in V}$ with
	\begin{align}
		\rho_{k\mid\Pa(k)}:
		\bigtimes_{\ell\in\Pa(k)}
		\left(
		\Ch(\ell)\cup\{\bot\}
		\right)
		&\rightarrow
		\mathbb Z_2\\
		(t_\ell)_{\ell\in\Pa(k)}
		&\mapsto
		\prod_{\ell\in\Pa(k)}[k=t_\ell]
		\,.
	\end{align}
	The bottom element~$(\bot)$ is a special element not contained in~$V$.
\end{definition}
These model parameters are such that each party~$\ell$ can select {\em at most\/} one of its children~$k\in\Ch(\ell)$.
If all parents of party~$k$ select~$k$, then party~$k$ receives one, and zero otherwise.
Note that this causal model is {\em faithful:}
Party~\mbox{$\ell\in\Pa(k)$} can signal a bit to party~$k$ whenever~$k$ is selected by all parents except~$\ell$, {\it i.e.,} with~$t_{\ell'}=k$ for all~$\ell'\in\Pa(k)\setminus\{\ell\}$, we have~\mbox{$\rho_{k\mid\Pa(k)}(k,\dots,k,t_\ell) = [k = t_\ell]$}.
The bottom element~$\bot$ in the output space of the parties is required to ensures faithfulness in the case~$|\Ch(\ell)|=1$.

Suppose we amend a siblings-on-cycles graph~$G$ with these model parameters.
If~$C$ is a directed cycle in the graph~$G$, then these model parameters interrupt the signal progression along the cycle~$C$, and therefore, no party can send a~signal to herself.
To see this, let~$i,j\in C$ be siblings and~$p\in\CoPa(\{i,j\})$.
If the common parent~$p$ selects party~$j=:t_p$, then~$\rho_{i\mid\Pa(i)}(t_p,t_{\Pa(i)\setminus\{p\}}) = 0$, {\it i.e.,\/}~party~$i$ receives the constant zero, and therefore, the parent of~$i$ along the cycle~$C$ cannot send a signal to~$i$.
Similarly, the signal progression along the cycle~$C$ is interrupted at party~$j$ if the common parent~$p$ selects~$i$.
This motivates the following conjecture:
\begin{conjecture}[Admissible causal structures]
	\label{conj:admissible}
	If~\mbox{$G=(V,E)$} is a siblings-on-cycles graph, then the causal structure~$G$ equipped with the model parameters of Definition~\ref{def:modelparameters} forms a {\em consistent\/} causal model.
\end{conjecture}
Note that the above is {\em not\/} a proof of this conjecture.
The problem is that the value~$t_p$ might depend on the interventions of other parties, and in particular, on the interventions of its children.

We have numerically tested this conjecture for all siblings-on-cycles graphs with up to six nodes.
For these tests, we employed the following recursive function~$\alpha$
\begin{align}
	\alpha:V\times V^* &\rightarrow \{0,1\}\\
	\begin{split}
		\left(k,\pi=(v_1,\dots,v_m)\right) &\mapsto \left[k\not\in\pi\right]\times\\
		\prod_{\ell\in\Pa(k)}\bigg[k = \mu_\ell\Big(&\alpha\big(\ell,(k,v_1,\dots,v_m)\big)\Big)\bigg]
		\,,
	\end{split}
\end{align}
where~$V^*$ is the set of all finite sequences of vertices.
We verified that~\mbox{$\hat i:=(i_k=\alpha(k))_{k\in V}$} is the fixed point of the map~$\omega\circ\mu$.
Note that~$\alpha(k)$ for~$k\in V$ is well-defined---the recursion terminates after finitely many invocations of~$\alpha$---because the directed graphs are finite, and whenever a~vertex is re-visited,~$\alpha$ returns~$0$.
For these tests, we wrote two C programs; one for generating siblings-on-cycles graph, and another to verify the fixed point.
The source code of these programs can be found in Ref.~\cite{code}.
We required 132 800MHz-CPU-seconds for the generation of the causal structures, and around 31 800MHz-CPU-hours for the verification of the admissibility of the model parameters.

Because every faithful and consistent classical-deterministic causal model can be lifted to a quantum one~\cite{baumeler2021}, this conjecture has as immediate consequence the completion of Theorem~\ref{thm:inadmissible} with its converse:
\begin{corollary}
	If Conjecture~\ref{conj:admissible} holds, then
	the causal structure~$G$ is {\em admissible\/} if and only if~$G$ is a {\em siblings-on-cycles\/} graph.
\end{corollary}

At present, we manage to prove a restricted form of Conjecture~\ref{conj:admissible} for a subset of siblings-on-cycles graphs, namely for all such graphs where every directed cycle is {\em induced.}
We call such graphs {\em chordless siblings-on-cycles graphs.}
\begin{theorem}[Admissible causal structures (chordless)]
	\label{thm:admissiblerestricted}
	If ~\mbox{$G=(V,E)$} is a chordless siblings-on-cycles graph, then the causal structure $G$ equipped with the model parameters of Definition~\ref{def:modelparameters} forms a {\em consistent\/} causal model.
\end{theorem}
Before we prove this statement, we need to establish the following lemma.
\begin{lemma}
	\label{lemma:subgraphs}
	Let~$G=(V,E)$ be a chordless siblings-on-cycles graph.
	If~$G$ contains a~directed cycle~$C$, then the induced graph~$G[V']$ with~$V'= \CoPa(C)\cup \Anc(\CoPa(C))$ is a chordless siblings-on-cycles graph,~$V'$ is non-empty, and a strict subset of~$V$.
\end{lemma}
\begin{proof}
	Let~$C=(c_0,c_1,\dots,c_{m-1})$ be a directed cycle in~$G$.
	This cycle has siblings~$c_i,c_j\in C$ with common parent~$p\in\Pa(c_i)\cap\Pa(c_j)$.
	Because~$C$ is induced, we have~\mbox{$p\not\in C$}~(see~Figure~\ref{fig:lemmaa}).
	\begin{figure}
		\centering
		\begin{subfigure}[b]{0.15\textwidth}
			\centering
			\begin{tikzpicture}
				\def\radius{.8}
				\def\angp{220-45}
				\def\angi{180-45}
				\def\angj{-45}
				\def\dd{3}
				\tikzstyle{vertex} = [fill,draw,circle,minimum size=\dd,inner sep=0pt]
				\node[vertex,label=\angp:{$p$}] (p) at (\angp:\radius) {};
				\node[vertex,label=\angi:{$c_i$}] (i) at (\angi:\radius) {};
				\node[vertex,label=\angj:{$c_j$}] (j) at (\angj:\radius) {};
				\draw[-stealth] (p) -- (j);
				\draw[-stealth] (\angp-\dd:\radius) arc(\angp-\dd:\angi+\dd:\radius);
				\draw[dashed,-stealth] (\angi-\dd:\radius) arc(\angi-\dd:\angj+\dd:\radius);
				\draw[dashed,-stealth] (\angj+360-\dd:\radius) arc(\angj+360-\dd:\angp+\dd:\radius);
			\end{tikzpicture}
			\caption{}
			\label{fig:lemmaa}
		\end{subfigure}
		\hfill
		\begin{subfigure}[b]{0.15\textwidth}
			\centering
			\begin{tikzpicture}
				\def\radius{.8}
				\def\angi{180-45}
				\def\angj{-45}
				\def\angl{270-45}
				\def\dd{3}
				\tikzstyle{vertex} = [fill,draw,circle,minimum size=\dd,inner sep=0pt]
				\node[vertex,label=-\angl:{$c_\ell$}] (l) at (\angl:\radius) {};
				\node[vertex,label=\angi:{$c_i$}] (i) at (\angi:\radius) {};
				\node[vertex,label=\angj:{$c_j$}] (j) at (\angj:\radius) {};
				\node[vertex,label=\angl:{$p$}] (p) at (\angl:2*\radius) {};
				\draw[dashed,-stealth] (\angi-\dd:\radius) arc(\angi-\dd:\angj+\dd:\radius);
				\draw[dashed,-stealth] (\angj+360-\dd:\radius) arc(\angj+360-\dd:\angl+\dd:\radius);
				\draw[dashed,-stealth] (\angl-\dd:\radius) arc(\angl-\dd:\angi+\dd:\radius);
				\draw[dashed,-stealth] (l) -- (p);
				\draw[-stealth] (p) to[out=\angl+180+45,in=\angi+90-30] (i);
				\draw[-stealth] (p) to[out=\angl+180-45,in=\angj+270+30] (j);
			\end{tikzpicture}
			\caption{}
			\label{fig:lemmab}
		\end{subfigure}
		\hfill
		\begin{subfigure}[b]{0.15\textwidth}
			\centering
			\begin{tikzpicture}
				\def\radius{.8}
				\def\angi{180-45}
				\def\angj{-45}
				\def\angl{270-45}
				\def\dd{3}
				\tikzstyle{vertex} = [fill,draw,circle,minimum size=\dd,inner sep=0pt]
				\tikzstyle{vertices} = [fill=white,draw,circle,minimum size=\dd,inner sep=0pt]
            	\node[vertex,label=\angi:{$c_i$}] (i) at (\angi:\radius) {};
				\node[vertex,label=\angj:{$c_j$}] (j) at (\angj:\radius) {};
			
				\node[vertices] (p) at (\angl:2*\radius) {$V'$};
				\draw[dashed,-stealth] (\angi-\dd:\radius) arc(\angi-\dd:\angj+\dd:\radius);
				\draw[dashed,-stealth] (\angj+360-\dd:\radius) arc(\angj+360-\dd:\angi+\dd:\radius);
				\draw[dashed,-stealth] (\angl-\dd:\radius) arc(\angl-\dd:\angi+\dd:\radius);
			
				\draw[-stealth] (p) to[out=\angl+180+45,in=\angi+90-30] (i);
				\draw[-stealth] (p) to[out=\angl+180-45,in=\angj+270+30] (j);
			\end{tikzpicture}
		
			\caption{}
			\label{fig:lemmac}
		\end{subfigure}
		\caption{Cases in which the directed cycle~$C$ appears in the graph. Solid lines represent edges, dashed ones paths. (a)~If the common parent~$p$ of~$c_i,c_j$ is an element in~$C$, then~$C$ is not induced.
			(b)~If there exists a~directed path from a node in~$C$ to~$p$, then the graph contains a~non-induced cycle.
			(c)~The only possibility is that there exists a set~$V'$ of nodes without paths from~$C$ to~$V'$.
		}
		\label{fig:lemma}
	\end{figure}
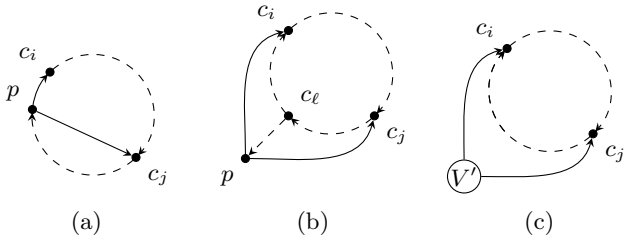
	More so,~$G$ contains no directed path~\mbox{$\pi=(c_\ell,\dots,p)$} with~$c_\ell\in C$.
	If there were such a~path and~$i\leq\ell<j$, then the directed cycle~$(p,c_j,\dots,c_i,\dots,c_\ell,\dots)$ would not be induced~(see~Figure~\ref{fig:lemmab}).
	A similar argument holds for~$\ell <i$ and~$j\leq\ell$.
	Since this holds for any pair of siblings on~$C$ and any common parent, we conclude that~$G$ does not contain any path from any node in~$C$ to any common parent~$p$~(see~Figure~\ref{fig:lemmac}).
	Now, the induced graph~$G[V']$ is a chordless siblings-on-cycles graph.
	The reason for this is that the nodes~$V'$ have the same incoming edges in~$G[V']$ and~$G$.
	Finally, note that~$V'$ is non-empty (it contains at least the node~$p$), and~$V'\subsetneq V$ because~$V'\cap C=\emptyset$.
\end{proof}
\begin{proof}[Proof of Theorem~\ref{thm:admissiblerestricted}]
	By Theorem~\ref{thm:cprocess}, the causal model is consistent if and only if for any choice of experiments~$\{\mu_k\}_{k\in\mathbb Z_n}$ there exists a {\em unique\/} fixed point of~$\omega\circ\mu$.
	Towards a contradiction, assume that~$\omega$ is {\em not\/} a classical-deterministic process.
	Therefore, and by Theorem~\ref{thm:equivalence}, let the experiments~$\{\mu_{k}\}_{k\in\mathbb Z_n}$ be such that~$\omega\circ\mu$ has at least two distinct fixed points,~$r$ and~$r'$.
	We start by observing two implications:
	\begin{align}
		&r_i \neq r'_i \implies \exists p \in \Pa(i): r_p\neq r'_p\,,\label{eq:fpparent}\\
		\begin{split}
			&r_i \neq r'_i \wedge r_j\neq r'_j \wedge  i\neq j\\
			&\qquad\implies\\
			&\forall p\in\CoPa(\{i,j\}): r_p \neq r'_p\,.\label{eq:fpcommonparent}
		\end{split}
	\end{align}
	For the first implication~\eqref{eq:fpparent}, suppose~$r$ and~$r'$ differ at position~$i$, and assume without loss of generality that~\mbox{$r_i=1$}.
	By the choice of model parameters, we have
	\begin{align}
		r_i &= 1=\prod_{p\in\Pa(i)} [ i = \mu_p(r_p)]\,,\\
		r'_i &= 0=\prod_{p\in\Pa(i)} [ i = \mu_p(r'_p)]\,.
	\end{align}
	This means that for {\em all\/} parents~$p\in\Pa(i)$ the identity~\mbox{$i=\mu_p(r_p)$} holds.
	But since~$r'_i=0$, there must exists a parent~$p\in\Pa(i)$ such that~$i\neq \mu_p(r'_p)$: The fixed points also differ on node~$p$.
	For the second implication~\eqref{eq:fpcommonparent}, additionally suppose that~$r_j\neq r'_j$ for a~node~$j$ different from~$i$.
	If~$i$ and~$j$ do not have common parents, then the implication trivially holds.
	For the alternative, let~$p\in\CoPa(\{i,j\})$ be an arbitrary common parent.
	Since~$r_i=1$ and~$p\in\Pa(i)$, we have~$i=\mu_p(r_p)$, which implies~$j\neq\mu_p(r_p)$, and therefore~$r_j=0$.
	This, in turn, implies~$r'_j=1$ and~$j=\mu_p(r'_p)$: The fixed points also differ for~$p$.

	Since the graph~$G$ is finite, we conclude from implication~\eqref{eq:fpparent} that~$G$ contains at least one directed cycle on which the fixed points differ.
	Now, let~$V_{\text{min}}\subseteq V$ be a~non-empty set of nodes with {\em minimal cardinality\/}~$|V_{\text{min}}|$ such that the induced graph~$G[V_{\text{min}}]$
	is a chordless \mbox{siblings-on-cycles} graph, and which contains a directed cycle~$C$ where the fixed points~$r$ and~$r'$ differ.
	Moreover, let~$i,j\in C$ be siblings on~$C$, and~$p\in\CoPa(\{i,j\})$ a~common parent.
	By Lemma~\ref{lemma:subgraphs}, there exists a {\em smaller\/} non-empty set of nodes~$V'=\CoPa(C)\cup\Anc(\CoPa(C))\subsetneq V_{\text{min}}$ such that the induced graph~$G[V']$ is a chordless siblings-on-cycles graph.
	Note that the common parent~$p$ and all its ancestors~$\Anc(p)$ are in the set~$V'$.
	Implication~\eqref{eq:fpcommonparent} now states that~$r_p\neq r'_p$, and implication~\eqref{eq:fpparent} that~$V'$ contains a directed cycle~$C'$ where the fixed points~$r$ and~$r'$ differ:
	The extremality of~$|V_{\text{min}}|$ is violated.
\end{proof}

\section{Causal correlations}
In this second part of our work, we derive consequences for the observable correlations from graph-theoretic properties of the causal structure.
Due to Theorem~\ref{thm:inadmissible}, we restrict ourselves to siblings-on-cycles graphs.

The following graph-theoretic lemma is helpful for our first result.
\begin{lemma}
	\label{lemma:source}
	If $G=(V,E)$ is a chordless siblings-on-cycles graph then $G$ contains a~source node, {\it i.e.,}~\mbox{$\exists k\in V:\deg_\text{in}(k)=0$}.
\end{lemma}
\begin{proof}
	We prove this lemma by contradiction and an extremality argument.
	Suppose~$G$ is a chordless siblings-on-cycles graph, but {\em without\/} source nodes.
	Now, let~\mbox{$V_{\text{min}}\subseteq V$} be a non-empty set of nodes with {\em minimal cardinality\/}~$|V_{\text{min}}|$ such that the induced graph~$G[V_{\text{min}}]$ has the same properties as~$G$, {\it i.e.,\/}~$G[V_{\text{min}}]$ is a chordless~siblings-on-cycles graph with no source nodes.
	Since~$G[V_{\text{min}}]$ is non-empty and every node has an incoming edge, this induced graph contains a directed cycle~$C=(c_0,c_1,\dots,c_{m-1})$.
	By Lemma~\ref{lemma:subgraphs}, however, there exists a {\em smaller\/} set of nodes~$V'$ such that the induced graph~$G[V']$ has the same properties as~$G$.
\end{proof}
 
We now relate causal structures with {\em causal\/} correlations:
\begin{theorem}[Causal correlations]
	\label{thm:causal}
	Let~$\omega$ be a classical-deterministic process with causal structure~$G=(V,E)$.
	If~$G$ is a chordless siblings-on-cycles graph, 
	then~$\omega$ always produces causal correlations, {\it i.e.,\/}~for all experiments~$\mu$ the correlations~$p(a\mid x)=\omega\star\mu^{a\mid x}$ are causal.
\end{theorem}
\begin{proof}
	By Lemma~\ref{lemma:source}, the graph~$G$ contains a source node~$k\in V$, {\it i.e.,\/}~$\Pa(k)=\emptyset$.
	Therefore, the~$k$-th component of~$\omega$ is a {\em constant\/} (function)~$\omega_k\in\mathcal I_k$ \mbox{($\omega_k:\emptyset\rightarrow\mathcal I_k$)},
	and we can express the correlations~$p(a\mid x)$ given by~Eq.~\eqref{eq:correlations} as
	\begin{align}
		\begin{split}
			&p(a\mid x)
			=
			\sum_{i_k,o_k}
			\left[ \omega_k = i_k \right]
			\left[ \left(a_k, o_k\right) = \mu_k\left(x_k, i_k\right) \right]\times\\
			&\sum_{i_{\setminus k},o_{\setminus k}}
			\left[ \omega_{\setminus k}(o) = i_{\setminus k} \right]
			\left[ \left(a_{\setminus k}, o_{\setminus k}\right) = \mu_{\setminus k}\left(x_{\setminus k}, i_{\setminus k}\right) \right]
		\end{split}\\
		\begin{split}
			&=
			\left[ a_k = \mu^0_k\left(x_k, \omega_k\right) \right]
			\times
			\bigg(
			\sum_{i_{\setminus k},o_{\setminus k}}
			\left[ \omega'_{\setminus k}\left(o_{\setminus k}\right) = i_{\setminus k} \right]\\
			&\quad
			\left[ \left(a_{\setminus k}, o_{\setminus k}\right) = \mu_{\setminus k}\left(x_{\setminus k}, i_{\setminus k}\right) \right]
			\bigg)
			\,,
			\label{eq:kalone}
		\end{split}
	\end{align}
	where~$\omega'_\ell(o_{\Pa(\ell)})$ is~$\omega_\ell(\mu_k^1(x_k,\omega_k),o_{\Pa(\ell)\setminus\{k\}})$ if~$k\in\Pa(\ell)$,
	and~$\omega_\ell(o_{\Pa(\ell)})$ otherwise.
	This~$\omega'_{\setminus k}$ is a {\em reduced process\/}~(see Lemma~\ref{lemma:reduced}).
	The former term of Eq.~\eqref{eq:kalone} describes the observed statistics~$p(a_k\mid x_k)$ for party~$k$, the latter has the form~$\omega'_{\setminus k}\star \mu^{a_{\setminus k}\mid x_{\setminus k}}_{\setminus k}$.
	Therefore, the correlations~$p(a\mid x)$ decompose as
	\begin{align}
		p(a\mid x) = p(a_k\mid x_k) p^{x_k}(a_{\setminus k}\mid x_{\setminus k})
		\,.
	\end{align}
	The causal structure~$G'$ of the reduced process~$\omega'_{\setminus k}$ remains a chordless siblings-on-cycles graph (no edges are introduced and the reduced process is a classical-deterministic process).
	By repeating this argument, we end at the decomposition of~$p(a\mid x)$ as in Definition~\ref{def:causal}.
\end{proof}

Theorem~\ref{thm:causal} and the proof thereof leads to the following conjecture for the quantum case:
\begin{conjecture}[Quantum causal correlations]
	\label{conj:causal}
	Let~$W$ be a quantum process with causal structure~$G=(V,E)$.
	If~$G$ is a chordless siblings-on-cycles graph,
	then~$W$ always produces causal correlations.
\end{conjecture}
It is well-known that the celebrated quantum switch mentioned in the introduction has a chordless siblings-on-cycles causal structure (see Figure~\ref{fig:introb}), and also that it does not violate causal inequalities~\cite{araujo2015,oreshkov2016}.

\section{Non-causal correlations}
Some causal structures imply that causal inequalities {\em can\/} be violated.
Before we state and prove these results, let us introduce the following causal game.

\begin{game*}[$G^n_{\mathcal S}$]
	\label{def:game}
	Consider a scenario with~$n$ parties~$\mathbb Z_n$ and let~$\mathcal S$ be a non-empty subset of~$\mathbb Z_n$.
	For each party~$k\in\mathbb Z_n$, the set of settings is~\mbox{$\mathcal X_k:=(\mathbb Z_n \cup \{\bot\})\times (\mathbb Z_2 \cup \{\bot\})$},
	and the set of observations is~$\mathcal A_k:=\mathbb Z_2$.
	A referee uniformly at random picks a party~$s\in\mathcal S$ and a bit~$b\in\mathbb Z_2$.
	Then, the referee distributes~$s$ to all parties in~$\mathcal S$,~$b$ to all parties in~$\mathcal S\setminus \{s\}$, and nothing ($\bot)$ to the remaining parties,
	{\it i.e.,\/}~the settings are~$x_s = (s,\bot)$,~$x_k = (s,b)$ for~$k\in\mathcal S\setminus\{s\}$, and~$x_\ell = (\bot,\bot)$ for~$\ell\in\mathbb Z_n\setminus\mathcal S$.
	The parties win the game~$G^n_{\mathcal S}$ whenever party~$s$ correctly guesses~$b$, {\it i.e.,\/} whenever~$a_s = b$.
\end{game*}
This causal game\footnote{Note that~$G^n_{\mathcal S}$ is a modification of the ``selective signaling game''~\cite{baumeler3parties}.} has a non-trivial upper bound on the winning probability for causal correlations.
\begin{theorem}[Causal inequality]
	\label{thm:causalinequality}
	If~$p(a\mid x)$ are~$n$-party causal correlations, then the winning probability of the game~$G^n_{\mathcal S}$ is bounded by
	\begin{align}
		\Pr[a_s = b] \leq 1 - \frac{1}{2|\mathcal S|}
		\,.
	\end{align}
\end{theorem}
\begin{proof}
	Suppose the correlations~$p(a\mid x)$ decompose as
	\begin{align}
		p(a\mid x) = p(a_k\mid x_k) p^{x_k}_{a_k}(a_{\setminus k}\mid x_{\setminus k})
		\,,
	\end{align}
	where~$p^{x_k}_{a_k}(a_{\setminus k}\mid x_{\setminus k})$ are~$n-1$-party causal correlations and~$k\in\mathcal S$.
	In the event that~$s=k$, which happens with probability~$1/|\mathcal S|$, the game is won with half probability~(the bit~$b$ is uniformly distributed and party~$s$ has no access to~$b$).
	In the event that~$s\neq k$, the game is won with probability at most one.
	Therefore, the winning probability is upper bounded by
	\begin{align}
		\frac{1}{|\mathcal S|}\left(\frac{1}{2}+|\mathcal S|-1\right) = 1 - \frac{1}{2|\mathcal S|}
		\,.
	\end{align}
	The same bound holds for any other decomposition of~\mbox{$p(a\mid x)$}, and therefore also for convex combinations thereof.
\end{proof}

This brings us to a graph-theoretic criterion which implies a violation of the above causal inequality.
\begin{theorem}[Non-causal correlations]
	\label{thm:noncausal}
	Let~$\omega$ be a classical-deterministic process with causal structure~\mbox{$G=(V,E)$}.
	If~$G$ contains a directed cycle~$C$ where all common parents are in~$C$, {\it i.e.,\/}~$\CoPa(C)\subseteq C$,
	then~$\omega$ produces non-causal correlations, {\it i.e.,\/}~there exists an experiment~$\mu$ such that the correlations~$p(a\mid x)=\omega\star\mu^{a\mid x}$ violate the causal inequality for the game~$G^{|V|}_{C}$, and therefore, the correlations~$p(a\mid x)$ are non-causal.
\end{theorem}
\begin{proof}
	Let the directed cycle~$C$ be~$(c_0,c_1,\dots,c_{m-1})$,~$\sigma$ such that~$s=c_\sigma$, and ~$s^-:=c_{\sigma -1\pmod m}$, and let~$D$ be the set~$\Pa(C)\setminus C$.
	The condition~$\CoPa(C)\subseteq C$ implies that every party~$d\in D$ has a {\em single\/} edge to a~party in~$C$.
	Therefore, and since the causal model is faithful, there exists an experiment~$\mu_D$ for all parties in~$D$ such that signals can be sent along the cycle~$C$.
	Again due to faithfulness, there exists an experiment~$\mu_{D'}$ for the parties in~$D':=(\Pa(s)\cap C) \setminus \{s^-\}$ such that party~$s^-$ can send a signal to~$s$.
	Finally, there exists an experiment for party~$s^-$ such that party~$s$ receives an encoding of~$b$ on the input space~$\mathcal I_s$ (the set~$\mathcal I_s$ contains at least two elements, but is not necessarily equal to~$\mathbb Z_2$).
	By implementing these experiments, an experiment~$\mu_s$ that decodes~$b$ from the input~$i_s$ for party~$s$, and an arbitrary experiment~$\mu_{\text{rem}}$ for all remaining parties,
	we obtain~$p(a\mid x)=\omega \star \mu^{(a\mid x)}$ such that~$a_s$ {\em deterministically\/} takes value~$b$, {\it i.e.,}~$\Pr[a_s=b]=1$.
	This violates the causal inequality of Theorem~\ref{thm:causalinequality}.
\end{proof}
With these results at hand, we can complete Figure~\ref{fig:admissiblethreepartystructures}:
Only graph number 16 leads to non-causal correlations.
This causal structure is actually the causal structure of the AF/BW process~\cite{af,baumeler2016space}---the first classical-deterministic process known to yield non-causal correlations.
An example of a classical-deterministic process~$\omega$ to deterministically win the game~$\mathcal G_{\mathcal S}^n$ is given in Ref.~\cite{baumeler2021}.
The processes described in that article have as causal structure the fully connected graph.
Such a graph clearly has a~Hamiltonian cycle~$\mathcal C_H$, and therefore, all common parents~$\CoPa(\mathcal C_H)$ are inside that cycle: Theorem~\ref{thm:noncausal} is applicable.

\section{Conclusion}
In the first part of this work, we characterized a set of causal structures for which any faithful causal model is {\em inconsistent.}
We, consequently, provided, given any directed graph $G$, a construction of classical-deterministic model parameters.
We conjecture---and prove a restricted form thereof---that any causal model with these model parameters and any causal structure not in the inconsistent set, is {\em consistent.}
This conjecture complements Theorem~\ref{thm:inadmissible}.
It would imply that there exists a causal model with causal structure~$G$ if and only if~$G$ is a siblings-on-cycles graph.
A direct consequence of this conjecture is that {\em quantum theory does not allow for more general causal connections when compared to classical theories.}

In the second part, we used this characterization to show that two sets of causal structures lead to either {\em causal\/} or {\em non-causal\/} correlations, in the classical-deterministic case.
Note that a {\em decisive\/} graph-theoretic criterion for \mbox{(non-)causal} correlations is impossible.
As a simple example, take the causal structure depicted in Figure~\ref{fig:example}.
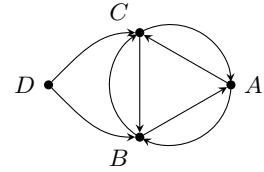
\begin{figure}
	\centering
	\begin{tikzpicture}
		\def\radius{.8}
		\def\anga{0}
		\def\angb{\anga+120+120}
		\def\angc{\anga+120}
		\def\dd{3}
		\tikzstyle{vertex} = [fill,draw,circle,minimum size=\dd,inner sep=0pt]
		\node[vertex,label=\anga:{$A$}] (A) at (\anga:\radius) {};
		\node[vertex,label=\angb:{$B$}] (B) at (\angb:\radius) {};
		\node[vertex,label=\angc:{$C$}] (C) at (\angc:\radius) {};
		\draw[-stealth] (\angb-\dd:\radius) arc(\angb-\dd:\angc+\dd:\radius);
		\draw[-stealth] (\angc-\dd:\radius) arc(\angc-\dd:\anga+\dd:\radius);
		\draw[-stealth] (\anga-\dd:\radius) arc(360+\anga-\dd:\angb+\dd:\radius);
		\draw[-stealth] (A) -- (C);
		\draw[-stealth] (C) -- (B);
		\draw[-stealth] (B) -- (A);
		\node[vertex,label=180+\anga:{$D$}] (D) at (180+\anga:2*\radius) {};
		\draw[-stealth] (D) to[out=\anga+45,in=\angc+90-30] (C);
		\draw[-stealth] (D) to[out=\anga-45,in=\angb+270+30] (B);
	\end{tikzpicture}
	\caption{Example of an admissible causal structure that does not satisfy the constraints of Theorem~\ref{thm:causal} or Theorem~\ref{thm:noncausal}.}
	\label{fig:example}
\end{figure}
If this causal structure is equipped with the model parameters from~Definition~\ref{def:modelparameters}, one obtains a consistent causal model that does {\em not\/} violate any causal inequality.
Intuitively, this follows because the role of the common parent~$D$ to the~$A,B,C$ cycle is equal to any other common parent.
So, effectively, the graph  is ``chordless.''
If, however, one uses these model parameters for the parties~$A,B,C$ only, and extends them with identity channels from~$\mathcal O_D$ to auxiliary input spaces of the parties~$B$ and~$C$,
then the resulting causal model is consistent and {\em violates\/} causal inequalities (in particular, the one of Theorem~\ref{thm:causalinequality} with~$\mathcal S=\{A,B,C\}$).
This holds because the induced graph with the nodes~$A,B,C$ is simply graph~16 in Figure~\ref{fig:admissiblethreepartystructures}.

Proving Conjecture~\ref{conj:admissible} would unlock many possibilities as it provides explicit examples of processes producing {\em non-causal\/} correlations for any number of parties and any admissible causal structure that contains a directed cycle where all  its common parents are part of.
It can be, therefore, a powerful tool to explore further non-causal correlations of different strength in the multi-party case.
Then again, Conjecture~\ref{conj:causal} would provide a better understanding of cyclic causal models that give rise to causal correlations.
The difficulty in proving this conjecture stems from the fact that the causal structure of a {\em reduced\/} quantum process might not necessarily be representable by a directed graph.

We briefly comment on how our results connect to various facets of other works.
Recently, several constructive bottom-up approaches to quantum processes with indefinite causal order have been explored (see, {\it e.g.,\/}~Refs.~\cite{wechs2021,vilasiniQPL,logicallconsistentcircuits}).
These works are contrasted by our complementary top-bottom approach.
Especially, to explore the connection to `routed circuits' and `sectorial decompositions'~\cite{vanrietvelde2020,vanrietvelde2021,ormrod2022} might be insightful, also in light of proving the conjectures presented.
Gogioso and Pinzani~\cite{pinzani2022} study causal order in a theory-independent fashion.
Theorem~\ref{thm:inadmissible}, however, imposes a restriction on causal orders for Markovian quantum processes.
For instance, the totality of two events in an indefinite causal order is inadmissible \cite{arajo2017,barrett2021,yokojima_2021}.
It is expected that Theorem~\ref{thm:inadmissible} imposes certain constraints on the ``joint decomposition'' of Markovian quantum causal orders.
Apadula, Bisio, and Perinotti~\cite{apadula2022} derive admissibility constraints for the composition of higher-order maps:
Two higher-order maps can be composed whenever no signaling loop is created.
If one wishes to compose the result with a third higher-order map, however, a more in-depth analysis is required.
In this light, our results can be understood as such an analysis of a slice within the hierarchy of higher-order computation.
Namely, we study the admissibility of {\em simultaneous\/} composition of channels (experiments) with processes.
Recently, Eftaxias, Weilenmann, and Colbeck~\cite{eftaxias2022} characterized the set of effects in the generalized probabilistic theory of ``box world.''
It is known that this theory admits effects that are {\em not\/} wirings~\cite{short2010}, and it turns out that these non-wiring effects are classical processes.
Our graph-theoretic criteria thus specifies the most general signaling structures of these effects.
Any operator with a cyclic signaling structure without siblings must therefore be excluded as effect.

A series of open questions---apart from proving the conjectures and the precise connection to related works---emerge.
A central question is to what extent our results are theory independent.
As suggested by Conjecture~\ref{conj:admissible}, the set of admissible quantum causal structures coincides with the classical one.
Does this potential coincidence extend to other theories?
Another question is how to embed admissible causal structures in space-time geometries of general relativity and quantum gravity~\cite{baumeler2019,tobar2020}, and how they can be achieved in a time-delocalized formulation~\cite{wechs2022} (note that some {\em fine-tuned\/} cyclic causal structures are known to be embeddable in Minkowski space-time~\cite{vilasini2022}).
Finally, this limitation on causal connections might unlock new information-processing protocols, {\it e.g.,\/}~in the presence of local quantum theory and classical communication without causal order~\cite{kunjwal2022}.

\noindent
{\bf Acknowledgments.}
We thank Luca Apadula, Marios Christodoulou, Timothée Hoffreumon, Ravi Kunjwal, Ognyan Oreshkov, Augustin Vanrietvelde, Vilasini Venkatesh, and the YIRG for fruitful discussions and helpful comments.
We thank two anonymous referees for helpful comments.
EET would like to thank Alexandra Elbakyan for providing access to the scientific literature.
We acknowledge support from the Austrian Science Fund (FWF) through ZK3 (Zukunftskolleg). ÄB also acknowledges support from the Austrian Science Fund (FWF) through BeyondC-F7103.
\"AB also acknowledges support from the Swiss National Science Foundation (SNF) through project~214808.

\bibliography{references.bib}

\end{document}